\documentclass[proceedings]{stacs}
\stacsheading{2010}{645-656}{Nancy, France}
\firstpageno{645}

\usepackage[english]{babel}
 \usepackage{times}
\usepackage[all]{xy}

\usepackage{stmaryrd}
	       {\begin{array}{@{\hspace{5em}}p{5em}p{50em}}}{\end{array}}

\newcommand{\C}{\mathcal{C}}
\newcommand{\NameCong}{\mathsf{NameCong}}
\newcommand{\Gm}{\mathsf{G}}
\newcommand{\Paste}{\mathsf{Paste}}
\newcommand{\Name}{\mathsf{Name}}
\newcommand{\Log}{\mathcal{L}}
\newcommand{\PV}{\mathsf{P}}
\newcommand{\Nom}{\mathsf{N}}
\newcommand{\hearts}{\heartsuit}
\newcommand{\Form}{\mathcal{F}}
\newcommand{\Set}{\mathsf{Set}}
\newcommand{\inv}{^{-1}}
\newcommand{\lsem}{\llbracket}
\newcommand{\rsem}{\rrbracket}
\newcommand{\Pow}{\mathcal{P}}
\newcommand{\CPow}{\mathcal{Q}}
\newcommand{\op}{^\mathit{op}}

\newcommand{\otto}{\leftrightarrow} 
\newcommand{\Hilb}{\mathsf{H}}

\newcommand{\entails}{\vdash}

\usepackage{amsmath}
\usepackage{amssymb}
\usepackage{xspace}
\usepackage{bbm}

\newcommand{\FA}{\mathfrak{A}}

\newcommand{\modelsOS}{\models}

\newcommand{\CondArrow}{\Rightarrow}

\newcommand{\CF}{\mathit{Cf}}

\theoremstyle{definition}
       \newtheorem{notn}{Notation}

\newcounter{blubber}
\newenvironment{sparenumerate}
{\begin{list}
  {\arabic{blubber}.}
  {\usecounter{blubber}
   \setlength{\leftmargin}{0pt}
    \setlength{\parsep}{0pt}
    \setlength{\itemindent}{4ex}
    \setlength{\itemsep}{2pt}
  }
}
{\end{list}}

\newcommand{\mi}[1]{\mathit{#1}}

\newlength{\croutw}
\newlength{\crouth}
\newcommand{\crossout}[1]%
        {\settowidth{\croutw}{$#1$}\settoheight{\crouth}{$#1$}#1%
        \hspace{-1.0\croutw}\raisebox{0.3\crouth}{\rule{\croutw}{0.1ex}}}

\newcommand{\commentout}[1]{\ignorespaces}

\newcommand{\infrule}[2]{\frac{#1}{#2}}

\newcommand{\bit}{\begin{itemize}}
\newcommand{\eit}{\end{itemize}}

        {\begin{enumerate}}%
        {\end{enumerate}}
        {\begin{enumerate}}%
        {\end{enumerate}}
        {\begin{enumerate}}%
        {\end{enumerate}}
        {\begin{enumerate}}%
        {\end{enumerate}}

\newcommand{\Cls}{\mathcal}
\newcommand{\Opname}{\mathrm}

\newcommand{\colim}{{\Opname{colim}\,}}

\newcommand{\lrule}[3]{(#1)\;\;\infrule{#2}{#3}}

\newcommand{\CT}{{\Cls T}}

\newcommand{\powerset}{{\mathcal P}}

\newcommand{\argument}{\_\!\_}

\newcommand{\Sem}[1]{{[\![#1]\!]}}

\newcommand{\modimpl}{\to}
\newcommand{\modiff}{\leftrightarrow}

\newcommand{\Pfin}{\mathcal P_{\mathit{fin}}}

        {\smallskip\par\noindent\qquad$\begin{array}{@{}l@{{}={}}l}}%
        {\end{array}$\par\noindent}

\newlength{\myboxwidth}
	{\end{center}\end{figure}}

\newcommand{\Lang}{\mathcal{L}}	
\newcommand{\FLang}{\mathcal{F}}

\newcommand{\gldiamond}[1]{\Diamond_{#1}}

\newcommand{\Nat}{{\mathbb{N}}}

\newcommand{\Rat}{{\mathbb{Q}}}

\newcommand{\PDist}{D_\omega}

\newcommand{\Bag}{\mathcal{B}}

\newcommand{\Prop}{\mathsf{Prop}}

\newcommand{\PSPACE}{\ensuremath{\mathit{PSPACE}}\xspace}
\newcommand{\Rules}{\mathcal{R}}

\newcommand{\CK}{\mathit{CK}}

\newcommand{\Ax}{\mathcal{A}}

\newcommand{\PLentails}{\entails_{\mi{PL}}}

\newcommand{\Space}{\mathcal{S}}
\newcommand{\osder}[2]{\entails_{#1}^{#2}}
\newcommand{\gentails}[1]{\entails^{#1}_g}

\newcommand{\eat}[1]{}

\begin{document}

\title{Named Models in Coalgebraic Hybrid Logic}
\author[DFKIUHB]{L. Schr{\"o}der}{Lutz
  Schr{\"o}der}
\address[DFKIUHB]{DFKI Bremen and Department of Computer Science,  
Universit\"at  Bremen}
\email{Lutz.Schroeder@dfki.de}
\author[IC]{D. Pattinson}{Dirk Pattinson}
\address[IC]{Department of Computing, 
Imperial College London}

\email{dirk@doc.ic.ac.uk}






\begin{abstract}
  Hybrid logic extends modal logic with support for reasoning about
  individual states, designated by so-called nominals. We study hybrid
  logic in the broad context of coalgebraic semantics, where
  Kripke frames are replaced with coalgebras for a given functor, thus
  covering a wide range of reasoning principles including, e.g.,
  probabilistic, graded, default, or coalitional
  operators. Specifically, we establish generic criteria for a given
  coalgebraic hybrid logic to admit named canonical models, with
  ensuing completeness proofs for pure extensions on the one hand, and
  for an extended hybrid language with local binding on the other.
  We instantiate our framework with a number of examples.
  Notably, we prove completeness of graded hybrid logic with local
  binding.
\end{abstract}

\maketitle             



\section*{Introduction}

\noindent Modal logics have traditionally played a central role in
Computer Science, appearing, e.g., in the guise of temporal logics,
program logics such as PDL, epistemic logics, and later as description
logics. The development of modal logics has seen extensions along (at
least) two axes: the enhancement of the expressive power of basic
(relational) modal logic on the one hand, and the continual extension,
beyond the purely relational realm, of the class of structures
described using modal logics on the other hand.  Hybrid logic falls
into the first category, extending modal logic with the ability to
reason about individual states in models.  This feature, originally
suggested by Prior and first studied in the context of tense logics
and PDL (see~\cite{BlackburnTzakova99} for references), is of
particular relevance in knowledge representation languages and as such
has found its way into modern description logics, where it is denoted
by the letter $\mathcal{O}$ in the standard naming
scheme~\cite{BaaderEA03}.

Extensions along the second axis -- semantics beyond Kripke structures
and neighbourhood models -- include various probabilistic modal
logics, interpreted over probabilistic transition systems, graded
modal logic over multigraphs~\cite{DAgostinoVisser02}, conditional
logics over selection function frames~\cite{Chellas80}, and coalition
logic~\cite{Pauly02}, interpreted over so-called game frames. As a
unifying semantic bracket covering all these logics and many further
ones, coalgebraic modal logic has emerged (\cite{CirsteaEA09} gives a
survey). The scope of coalgebraic modal logic has recently been
expanded to encompass nominals; we refer to the arising class of
logics as \emph{coalgebraic hybrid logics}. Existing results include a
finite model result, an internalized tableaux calculus, and generic
\PSPACE upper bounds, but are so far limited to logics that exclude
frame conditions and local binding~\cite{MyersEA09}. What is missing
from this picture technically is a theory of \emph{named canonical
  models}~\cite{BlackburnTzakova99}. Named canonical models yield not
only strong completeness of the basic hybrid logic, but also
completeness of \emph{pure extensions}, defined by axioms that do not
contain propositional variables (but may contain nominals; e.g.\ in
Kripke semantics, the pure axiom $\Diamond\Diamond i \modimpl\Diamond
i$, with $i$ a \emph{nominal}, defines transitive frames).  Moreover,
named canonical models establish completeness for an extended hybrid
language with a local binding operator $\downarrow x.\,\phi(x)$, read
as ``the current state $x$ satisfies $\phi(x)$''. Both pure extensions
and the language with $\downarrow$ (not addressed in~\cite{MyersEA09})
are, in general, undecidable~\cite{Areces99} (it should be noted,
however, that fragments of the language with $\downarrow$ over Kripke
frames are decidable and as such play a role, e.g., in conjunctive
query answering in description logic~\cite{GlimmEA06}). As a
consequence, completeness of pure extensions and local binding is the
best we can hope for -- it establishes recursive enumerability of the
set of valid formulas, and it enables automated reasoning, if not
decision procedures.

Specifically, we establish two separate criteria for the existence of
named models. Although these criteria are (in all likelihood
necessarily) less widely applicable than some previous coalgebraic
results including those of~\cite{MyersEA09}, the generic results allow
us to establish new completeness results for a wide variety of logics;
in particular, we prove strong completeness of graded hybrid logic,
and ultimately an extension of the description logic $\mathcal{SHOQ}$,
with the $\downarrow$ binder over a wide variety of frame classes.

\section{Coalgebraic Hybrid Logic}

\noindent To make our treatment parametric in the syntax, we fix a modal
similarity type $\Lambda$ consisting of modal operators with
associated arities throughout.  For given countably infinite and
disjoint sets $\PV$ of propositional variables and $\Nom$ of nominals,
the set $\Form(\Lambda)$ of \emph{hybrid $\Lambda$-formulas} is given
by the grammar
\[ \Form(\Lambda) \ni \phi, \psi ::=  p \mid i \mid
  \phi \land \psi \mid \lnot \phi \mid \hearts(\phi_1, \dots, \phi_n)
	\mid @_i \phi
\]
where $p \in \PV$, $i \in \Nom$ and $\hearts \in \Lambda$ is an
$n$-ary modal operator. (Alternatively, we could regard
\emph{propositional variables} as nullary modal operators, thus
avoiding their explicit mention altogether. We keep them explicit
here, following standard practice in modal logic, as we have to deal
with valuations anyway due to the presence of nominals.)
We use the standard definitions for the other propositional
connectives $\modimpl,\modiff,\lor$. The set of nominals occurring in
a formula $\phi$ is denoted by $\Nom(\phi)$, similarly for sets of
formulas. A formula of the form $@_i\phi$ is called an
\emph{$@$-formula}.  Semantically, nominals $i$ denote individual
states in a model, and $@_i\phi$ stipulates that $\phi$ holds at state
$i$.

To reflect parametricity also semantically, we equip hybrid logics
with a \emph{coalgebraic semantics} extending the standard coalgebraic
semantics of modal logics~\cite{Pattinson03}: we fix throughout a
\emph{$\Lambda$-structure} consisting of an endofunctor $T: \Set \to
\Set$ on the category of sets, together with an assignment of an
\emph{$n$-ary predicate lifting} $\lsem \hearts \rsem$ to every
$n$-ary modal operator $\hearts \in \Lambda$, i.e.\ a set-indexed
family of mappings $(\lsem \hearts \rsem_X: \Pow(X)^n \to \Pow(TX))_{X
  \in \Set}$ that satisfies
\[ \lsem \hearts \rsem_X \circ (f\inv)^n = (Tf)\inv \circ \lsem
\hearts \rsem_Y \] for all $f: X \to Y$. In categorical terms,
$\Sem{\hearts}$ is a natural transformation $\CPow^n \to \CPow \circ
T\op$ where $\CPow: \Set\op \to \Set$ is the contravariant powerset
functor.

In this setting, $T$-coalgebras play the roles of \emph{frames}. A
\emph{$T$-coalgebra} is a pair $(C, \gamma)$ where $C$ is a set of
\emph{states} and $\gamma: C \to TC$ is the \emph{transition
  function}. When $\gamma$ is clear from the context, we refer to $(C,
\gamma)$ just as $C$. A \emph{(hybrid) $T$-model} $M = (C, \gamma,
V)$ consists of a $T$-coalgebra $(C,\gamma)$ together with a
\emph{hybrid valuation} $V$, i.e. a map $\PV \cup \Nom \to \Pow(C)$
that assigns singleton sets to all nominals $i \in \Nom$.  We say that
$M$ is \emph{based} on the frame $(C, \gamma)$. The singleton set
$V(i)$ is tacitly identified with its unique element.

The semantics of $\FLang(\Lambda)$ is a satisfaction relation
$\models$ between states $c\in C$ in hybrid $T$-models $M = (C,
\gamma, V)$ and formulas $\phi \in \Form(\Lambda)$, inductively
defined as follows. For $x \in \Nom \cup \PV$ and $i \in \Nom$, 
\begin{equation*}
  M, c \models x \mbox{ iff } c \in V(x)\qquad\text{ and }\qquad
  M, c \models @_i \phi \mbox{ iff } M,V(i) \models \phi.
\end{equation*}
Modal operators are interpreted using their associated predicate
liftings, that is,
\begin{equation*}
  M, c \models \hearts(\phi_1, \dots, \phi_n) \iff
  \gamma(c) \in \lsem \hearts \rsem_C (\lsem \phi_1 \rsem_M, \dots,
  \lsem \phi_n \rsem_M)
\end{equation*}
where $\hearts \in \Lambda$ is $n$-ary and $\lsem \phi \rsem_M =
\lbrace c \in C \mid M,c \models \phi \rbrace$ denotes the truth-set
of $\phi$ relative to $M$.  
We write $M\models\phi$ if $M,c\models\phi$ for all $c\in C$. For a
set $\Phi \subseteq \Form(\Lambda)$ of formulas, we write $M,c \models
\Phi$ if $M, c \models \phi$ for all $\phi \in \Phi$, and
$M\models\Phi$ if $M\models\phi$ for all $\phi\in\Phi$. We say that
$\Phi$ is \emph{satisfiable} in a model $M$ if there exists a state
$c$ in $M$ such that $M,c\models\Phi$. If $\Ax \subseteq
\Form(\Lambda)$ is a set of axioms, also referred to as \emph{frame
  conditions}, a frame $(C, \gamma)$ is an \emph{$\Ax$-frame} if $(C,
\gamma, V) \models \phi$ for all hybrid valuations $V$ and all $\phi
\in \Ax$, and a model is an \emph{$\Ax$-model} if it is based on an
$\Ax$-frame. A frame condition is \emph{pure} if it does not contain
any propositional variables (it may however contain nominals). We
recall notation from earlier work:
\begin{notn} As usual, application of substitutions
  $\sigma:\PV\to\Form(\Lambda)$ to formulas $\phi$ is denoted
  $\phi\sigma$.  For a set $\Sigma$ of formulas and a set $O$ of
  operators, we write $O\Sigma$ or $O(\Sigma)$ for the set of formulas
  arising by prefixing elements of $\Sigma$ with an operator from $O$;
  e.g.\ $\Lambda(\Sigma)= \lbrace \hearts(\phi_1, \dots, \phi_n) \mid
  \hearts\in\Lambda \textrm{ $n$-ary}, \phi_1, \dots, \phi_n \in
  \Sigma \rbrace$ and $@\Sigma:=\{@_i\mid
  i\in\Nom\}(\Sigma)=\{@_i\phi\mid
  i\in\Nom,\phi\in\Sigma\}$. Moreover, $\Prop(Z)$ denotes the set of
  propositional combinations of elements of some set $Z$.  For
  $\phi\in\Prop(Z)$, we write
  $X, \tau \models \phi$ if $\phi$ evaluates to $\top$ in the boolean
  algebra $\Pow(X)$ under a valuation $\tau: Z \to \Pow(X)$. For
  $\psi\in\Prop(\Lambda(Z))$, the interpretation $\lsem \psi
  \rsem_{TX, \tau}$ of $\psi$ in the boolean algebra $\Pow(TX)$ under
  $\tau$ is the inductive extension of the assignment $\lsem \hearts
  (p_1, \dots, p_n ) \rsem_{TX, \tau} = \lsem \hearts \rsem_X(
  \tau(p_1), \dots, \tau(p_n))$.  We write $TX, \tau \models \psi$ if
  $\lsem \psi \rsem_{TX, \tau} = TX$, and $t\models_{TX,\tau}\psi$ if
  $t\in\lsem \psi \rsem_{TX, \tau}$. A set of formulas $\Xi \subseteq
  \Prop(\Lambda(Z))$ is \emph{one-step satisfiable} w.r.t.\ $\tau$ if
  $\bigcap_{\phi \in \Xi} \lsem \phi \rsem_{TX, \tau} \neq
  \emptyset$. We occasionally apply this notation to sets $Z\subseteq
  \Pow(X)$ with $\tau$ being just inclusion, in which case mention of
  $\tau$ is suppressed.
\end{notn}
\noindent In the sequel, we will be interested in both \emph{local}
and \emph{global} semantic consequence, where local consequence refers
to satisfaction in a single state and global consequence to
satisfaction in entire models. In fact, we consider local reasoning
under global assumptions: given a set $\Phi\subseteq\Form(\Lambda)$ of
global assumptions (a \emph{TBox} in description logic terminology)
and a class $\C$ of models, we say that \emph{$\phi$ is a local
  consequence of $\Psi$ under global assumptions $\Phi$ for
  $\C$-models}, in symbols $\Phi;\Psi \models^\C \phi$, if for all
$M\in\C$ such that $M\models\Phi$, $M, c \models \phi$ whenever $M, c
\models \Psi$ (here, both $\Phi$ and $\Psi$ are sets of arbitrary
formulas, in particular not subject to any restrictions on the nesting
depth of modal operators). The standard notions of local and global
consequence are regained from this general definition by taking $\Phi$
or $\Psi$ to be empty, respectively.

The distinguishing feature of the coalgebraic approach to hybrid and
modal logics is the parametricity in both the logical language and the
notion of frame: concrete instantiations of the general framework, in
other words a choice of modal operators $\Lambda$ and a
$\Lambda$-structure $T$, capture the syntax and semantics of a wide
range of modal logics, as witnessed by the following examples.

\begin{exas}\label{expl:logics}
\begin{sparenumerate}
\item The hybrid version of the modal logic $K$, \emph{hybrid $K$} for
  short, has a single unary modal operator $\Box$, interpreted over
  the structure consisting of the powerset functor $\Pow$ (which takes
  a set $X$ to its powerset $\Pow(X)$) and the predicate lifting
  $\lsem \Box \rsem_X(A) = \lbrace B \in \Pow(X) \mid B \subseteq A
  \rbrace$. It is clear that $\Pow$-coalgebras $(C, \gamma: C \to
  \Pow(C))$ are in 1-1 correspondence with Kripke frames, and that the
  coalgebraic definition of satisfaction specializes to the usual
  semantics of the box operator.
\item \emph{Graded hybrid logic} has modal operators $\gldiamond{k}$
  `in more than $k$ successors, it holds that'. It is interpreted over
  the functor $\Bag$ that takes a set $X$ to the set
  $\Bag(X)=X\to\Nat\cup\{\infty\}$ of multisets over $X$ by
  $\Sem{\gldiamond{k}}_X(A)=\{B\in\Bag(X)\mid \sum_{x\in A}B(x) >
  k\}$. This captures the semantics of graded modalities over
  \emph{multigraphs}~\cite{DAgostinoVisser02}, which are precisely the
  $\Bag$-coalgebras. A more general set of operators is that of
  \emph{Presburger logic}~\cite{DemriLugiez06}, which admits integer
  linear inequalities $\sum a_i\cdot\#(\phi_i)\ge k$ among formulas.
  Unlike in the purely modal case~\cite{Schroder07}, hybrid multigraph
  semantics visibly differs from the more standard Kripke semantics of
  graded modalities, as the latter validates all formulas
  $\lnot\gldiamond{1}i$, $i\in\Nom$.  However, both semantics agree if
  we additionally stipulate $\lnot\gldiamond{1}i$ as a global (pure)
  axiom.  Thus, our completeness 
  results for multigraph semantics derived below do transfer to Kripke
  semantics. In particular they apply to many description logics,
  which commonly feature both nominals and graded modal operators in
  the guise of \emph{qualified number restrictions}.  
\item \emph{Hybrid $\CK$}, the hybrid extension of the basic
  conditional logic $\CK$, has a single binary modal operator
  $\CondArrow$, written in infix notation. Hybrid $\CK$ is interpreted
  over the functor $\CF$ that maps a set $X$ to the set
  $\Pow(X)\to\Pow(X)$, whose coalgebras are selection function
  models~\cite{Chellas80}, by putting $\Sem{\CondArrow}_X(A,B)=
  \{f:\powerset(X)\to\powerset(X)\mid f(A)\subseteq B\}$.
\item\label{item:classical} \emph{Classical hybrid logic} (the hybrid
  version of the logic $E$ of neighbourhood frames, referred to as
  (the minimal) classical modal logic in~\cite{Chellas80}) has a
  single, unary modal operator $\Box$ and is interpreted over
  \emph{neighbourhood frames}, that is, coalgebras for the functor $NX
  = \Pow(\Pow(X))$ (more precisely, the double contravariant powerset
  functor). The semantics of classical modal logic is defined by the
  lifting $\lsem \Box \rsem_X (A) = \lbrace S \in N X \mid A \in S
  \rbrace$. \emph{Monotone hybrid logic} has the same similarity type,
  but is interpreted over upwards closed neighbourhood frames, or
  coalgebras for the functor $MX = \lbrace S \in N X \mid S \mbox{
    upwards closed} \rbrace$ where upwards closure refers to subset
  inclusion.
\item The syntax  of coalition logic over a set $N$ of agents is
given by the similarity type $\lbrace [C] \mid C \subseteq N
\rbrace$, and the operator $[C]$ reads as ``coalition $C$ has a
joint strategy to enforce \dots''. The formulas of (hybrid)
coalition logic are interpreted over game frames, i.e.,
coalgebras for the functor
\[ \Gm (X) = \lbrace (f, (S_i)_{i \in N}) \mid\textstyle \prod_{i \in N} S_i
\neq \emptyset, f: \prod_{i \in N} S_i \to X \rbrace
\]
(a class-valued functor, technically speaking, which however does not
cause problems).  The semantics arises via the liftings
\[
  \lsem [C] \rsem_X(A) = \lbrace (f, (S_i)_{i \in N}) \in \Gm (X)
	  \mid 
		  \exists (s_i)_{i \in C} \forall (s_i)_{i \in N \setminus C} (
	  f((s_i)_{i \in N}) \in A \rbrace.
\]
\end{sparenumerate}
\end{exas}


\noindent We proceed to present a Hilbert-style proof system for
coalgebraic hybrid logics, which we prove to be sound and strongly
complete. This requires that the logic at hand satisfies certain
coherence conditions between the axiomatization and the semantics ---
in fact the \emph{same} conditions as in the purely modal case, which
are easily verified \emph{local} properties that can be verified
without reference to $T$-models and are already known to hold for a
large variety of logics~\cite{Pattinson03,Schroder07}.

Proof systems for coalgebraic logics are most conveniently described
in terms of one-step rules, as follows.

\begin{defi}
A \emph{one-step rule} over $\Lambda$ is 
a rule $\phi/\psi$ where $\phi\in \Prop(\PV)$ and
$\psi\in\Prop(\Lambda(\PV))$ (in fact, $\psi$ may be restricted to be
a disjunctive clause, which however is not relevant here).  The rule
$\phi/ \psi$ is \emph{one-step sound} if $TX, \tau \models \psi$
whenever $X, \tau \models \phi$ for a valuation $\tau:\PV\to\Pow(X)$.
Given a set $\Rules$ of one-step rules and a valuation $\tau: \PV \to
\Pow(X)$, a set $\Xi \subseteq \Prop(\Lambda(\PV))$ is \emph{one-step
  consistent}~\cite{SchroderPattinson08b} if the set $\Xi\cup\lbrace \psi
\sigma \mid\sigma:\PV\to\Prop(\PV);\phi/ \psi \in
\Rules; X, \tau \models \phi \sigma\rbrace$ is
propositionally consistent.
\end{defi}
\noindent One-step sound rules are sound, and we will assume
one-step soundness tacitly in the sequel.
Completeness hinges on
variants of the notion of one-step completeness~\cite{Schroder07},
which we define further below.  As the notion of one-step rule does
not involve hybrid features, suitable rule sets can just be inherited
from the corresponding modal systems; for graded logics, conditional
logics, and many others, such rule sets are found, e.g.,
in~\cite{SchroderPattinson08,SchroderPattinson07}. We recall that the
one-step complete rule set for (hybrid) $K$ consists of the rules
  \begin{equation*}
	  \infrule{a}{\Box a} \qquad \infrule{a \land b \to c}{\Box a
		\land \Box b \to \Box c}\;.
  \end{equation*}
  A set $\Rules$ of one-step rules now gives rise to a Hilbert system
  $\Log\Rules$ by adjoining propositional tautologies and the hybrid
  axioms, and closing under modus ponens, rule application, and
  $@$-necessitation. Formally, we write $\Phi \entails_{\Log\Rules}
  \phi$ for a set $\Phi$ of formulas, the \emph{global assumptions}
  (or the \emph{TBox}), and a formula $\phi$ if $\phi$ is contained in
  the smallest set that
\begin{itemize}
\item contains $\Phi$ and all instances of propositional tautologies
\item contains all instances of $@$-introduction $i \land \phi \to @_i \phi$ and make-or-break
\begin{equation*}
(\mathrm{mob})   \quad @_i p \to (\hearts(q_1, \dots, q_n)
  \otto \hearts(@_i p \land q_1, \dots, @_i p \land q_n))
\end{equation*}
together with all instances of the axioms $\lnot @_i\bot$, $\lnot @_i
\phi \otto @_i \lnot \phi$,
$@_i(\phi\land\psi)\modiff(@_i\phi\land@_i\psi))$, $@_ii$, $@_i j
\otto @_j i$, $@_i k \land @_j p \to @_i p$; and
\item is closed under instances of $@$-generalization $p  / @_i
p$, instances of rules in $\Rules$, and modus ponens.
\end{itemize}
\noindent The second group of axioms ensures that $i\sim j:\equiv
@_ij$ defines an equivalence relation on nominals and that $@_i$
distributes over propositional connectives.  The (mob) axiom captures
the fact that the truth set of an $@$-formula is either empty or the
whole model; in the case of hybrid $K$, it is equivalent to the
standard back axiom $@_i\phi\modimpl\Box @_i\phi$.

We write $\Phi; \Psi \entails_{\Lang\Rules} \phi$ if there are
$\psi_1, \dots, \psi_n\in\Psi$ such that
$\Phi\entails_{\Lang\Rules}\psi_1 \land \dots \land
\psi_n\modimpl\phi$. That is, $\Phi; \Psi \entails_{\Lang\Rules }\phi$
if there is a proof of $\phi$ from global assumptions $\Phi$ that
additionally assumes $\Psi$ \emph{locally}. As we assume that all
one-step rules in $\Rules$ are one-step sound, soundness for both
local and global consequence is immediate: we have $\Phi; \Psi
\models^\C \phi$ (for $\C$ the class of all models) whenever $\Phi;
\Psi \entails_{\Lang\Rules} \phi$.
In~\cite{MyersEA09}, a criterion has been given for $\Log\Rules$ to be
\emph{weakly complete}, i.e.\ complete for the case where both the
TBox $\Phi$ and the set $\Psi$ of local assumptions are empty. Here,
we extend this result to combined strong \emph{global} and strong
\emph{local} completeness, i.e.\ to cover both an arbitrary TBox and
an arbitrary set of local assumptions, even if $\Log\Rules$ is
extended with pure frame conditions and local binding.

\section{Strong Completeness of Pure Extensions}

\noindent Pure completeness is a celebrated result in hybrid logic
\cite[Chapter 7.3]{BlackburnEA01}. In a nutshell, adding pure axioms
to an already complete proof system for the hybrid extension of the
modal logic $K$ (Example \ref{expl:logics}), one retains completeness
with respect to the class of frames that satisfy the additional
axioms. In contrast to arbitrary modal axioms, pure axioms do not
contain propositional variables, and therefore define -- in the
classical setting of hybrid $K$ -- first-order frame conditions. Here,
we show that the same theorem is valid for a much larger class of
logics, namely all coalgebraic hybrid logics satisfying one of two
suitable sets of conditions. For the sake of readability, we restrict
the technical development (not the examples) to the case of unary
operators from now on until Section~\ref{sec:bounded}.

\begin{defi}
  If $\Ax$ is a set of pure formulas and $\Rules$ is a set of one-step
  rules, we write $\Phi;\Psi \entails_{\Log\Rules\Ax + \Name} \phi$ if
  there are $\psi_1, \dots, \psi_n \in \Psi$ such that $\psi_1 \land
  \dots \psi_n \to \phi$ is $\Lang\Rules$-derivable from assumptions
  in $\Phi$ where additionally all substitution instances of axioms in
  $\Ax$ and the rule
\[ (\Name) \frac{i \to \phi}{\phi} (i \notin \Nom(\phi)) \]
may be used in deductions.  As before, we write $\Phi
\entails_{\Lang\Rules\Ax + \Name} \phi$ if $\Phi; \emptyset
\entails_{\Lang\Rules\Ax + \Name} \phi$.
\end{defi}

\noindent In the above system, the rule $(\mathsf{Name'})\;\,
@_i\phi/\phi$ ($i\notin\Nom(\phi)$ and the rule
\begin{equation*}
  \lrule{\NameCong}{@_j(\phi\modiff\psi)}{\hearts\phi\modiff \hearts\psi}
  \;(j\notin\Nom(\phi,\psi))
\end{equation*}
are derivable. The system is clearly sound for both global and local
consequence over $\Ax$-models in the same sense that $\Log\Rules$ is
sound over $T$-models.

\begin{defi}
Let $\Ax \subseteq \Form(\Lambda)$ be a set of pure axioms, and let
$\Phi \subseteq \Form(\Lambda)$ be a TBox.
A set $\Psi \subseteq \Form(\Lambda)$ is \emph{$(\Log\Rules\Ax +
  \Name)$-$\Phi$-inconsistent} if there are $\psi_1, \dots, \psi_n \in
\Psi$ such that $\Phi \entails_{\Log\Rules\Ax + \Name} \neg(\psi_1
\land \dots \land \psi_n)$. Otherwise, $\Psi$ is \emph{$(\Log\Rules\Ax
  + \Name)$-$\Phi$-consistent}.
A subset of $@\FLang(\Lambda)$, i.e.\ a set of $@$-formulas, is called
an \emph{ABox} (again borrowing terminology from description logic). A
\emph{maximally $(\Log\Rules\Ax + \Name)$-$\Phi$-consistent ABox} is a
maximal element $K$ among the $(\Log\Rules\Ax +
\Name)$-$\Phi$-consistent ABoxes, ordered by inclusion.
For such a $K$, we write $S_K=\{K_i\mid i\in\Nom\}$, where $K_i=\{\phi
\in \Form(\Lambda) \mid @_i\phi\in K\}$, and put
$V_K(i)=\{K_i\}=\{K_j\in S_K\mid i\in K_j\}$.
\end{defi}
\noindent
For the construction of a named model, we now fix a maximally
$(\Log\Rules\Ax + \Name)$-$\Phi$-consistent ABox $K$. Later, we will
take $K$ to be a maximally consistent extension of a given set $\Phi$
of formulas, where we may assume, thanks to the rule
$(\mathsf{Name'})$, that $\Phi\subseteq @\Form(\Lambda)$. We note the
following trivial facts:
\begin{lem} \label{lemma:containment} We have $\psi \sigma \in K_i$
  for all $\psi \in \Ax$ and all substitutions $\sigma$, and moreover
  $K \cup \Phi \subseteq K_i$.
\end{lem}
\noindent
Our goal is the construction of named canonical models in the
following sense:
\begin{defi}
  A \emph{named canonical $K$-model} is a model $(S_K,\gamma,V_K)$
  such that
  \begin{equation*}
    \gamma(K_i)\in\Sem{\hearts}\hat\phi\quad
    \textrm{iff}\quad
    \hearts\phi\in K_i
  \end{equation*}
  for every nominal $i$, where $\hat\phi=\{K_j\in S_K\mid\phi\in
  K_j\}$.
\end{defi}
\noindent It is clear that named canonical models are countable, as
there are only countably many nominals.
%
\begin{lem}[Truth lemma for named canonical models]\label{lemma:ntl}
  If $M=(S_K,\gamma,V_K)$ is a named canonical $K$-model and $\phi$ is a
  hybrid formula, then for every $K_i\in S_K$,
  \begin{equation*}
    M,K_i \models\phi\quad\textrm{iff}\quad\phi\in K_i.
  \end{equation*}
  Hence, $M\models\Phi$, and $M$ is an $\Ax$-model.
\end{lem}

\noindent The last clause of the truth lemma follows from
Lemma~\ref{lemma:containment}, the crucial point being that
satisfaction of all substitution instances of $\Ax$ implies frame
satisfaction of $\Ax$ because every state in the model is denoted by
some nominal. We now establish two criteria for the existence of named
canonical models. The first criterion assumes a stronger form of
one-step completeness than the second, which instead demands that the
modalities are \emph{bounded}.
\subsection{Pure Completeness for Strongly One-Step Complete Logics}
\label{sec:sosc}
\noindent
The construction of named models hinges on the following notion of
pastedness, which assures that nominals interact correctly across the
whole model. For the rest of the section, we fix a one-step complete
rule set $\Rules$, a set $\Ax$ of pure axioms, and a set $\Phi
\subseteq \Form(\Lambda)$ of global assumptions, and we write
`consistent' instead of `$(\Log\Rules\Ax +
\Name)$-$\Phi$--consistent'.

\begin{defi}\label{def:paste-zero}
  An ABox $K$ is \emph{$0$-pasted} if whenever $@_j(\phi\modiff\psi)
  \in K$ for all nominals $j$, then $@_i(\hearts\phi\modiff \hearts\psi)\in K$ for
  all nominals~$i$.
\end{defi}
\noindent It is clear that $K$ can induce a named model only if $K$ is
$0$-pasted. The construction of pasted ABoxes requires a Henkin-like
extension of the logical language by adding new
nominals. \emph{Generally, we denote by $\Form(\Lambda)^+$ an extended
  language with countably many new nominals not appearing in
  $\Form(\Lambda)$.} We note the fact (slightly glossed over in the
literature) that this extension is conservative:
\begin{lem}\label{lem:cons}
  If $\Psi\subseteq\Form(\Lambda)$ is consistent, then $\Psi$ remains
  consistent in $\Form(\Lambda)^+$.
\end{lem}

\begin{lem}[Extended Lindenbaum lemma for $0$-Pasted Sets]\label{lemma:ell-0}
  If $\Psi\subseteq \FLang(\Lambda)$ is consistent, then there exists
  a $0$-pasted maximally consistent ABox $K\subseteq
  @\Form(\Lambda)^+$ and a nominal $i$ in $\Form(\Lambda)^+$ such that
  $@_i\Psi \subseteq K$.
\end{lem}
\noindent
(The proof of the above version of the Lindenbaum lemma uses
Lemma~\ref{lem:cons}, and exploits the $\Name'$ rule to introduce the
nominal $i$.)  As we are aiming for strong completeness results,
(weak) one-step completeness as employed in weak completeness proofs
using \emph{finite} models~\cite{MyersEA09,Schroder07} is no longer
adequate. Accordingly, our first criterion assumes a stronger
condition:
\begin{defi}
  A rule set $\Rules$ is \emph{strongly one-step complete} if for
  every set $X$, every one-step consistent subset of
  $\Prop(\Lambda(\Pow(X)))$ is one-step satisfiable.
\end{defi}
%
%
\begin{lem}[Named existence lemma, Version 1] \label{lemma:nex-1}
  If $K$ is $0$-pasted and $\Rules$ is strongly one-step complete, then
  there exists a named canonical $K$-model.
\end{lem}
\noindent
In summary, we have:
\begin{thm}\label{thm:pure-ext-mcs}
  If $\Rules$ is strongly one-step complete, then every extension of
  $\Lang\Rules$ by pure axioms is both globally and locally strongly
  complete over countable hybrid models when equipped with the $\Name$
  rule.  That is, if $\Phi, \Psi \subseteq \Form(\Lambda)$ and
  $\phi\in\Form(\Lambda)$, then $\Phi; \Psi \entails_{\Lang\Rules\Ax +
    \Name} \phi$ whenever $\Phi; \Psi \models^\C \phi$, where $\C$ is
  the class of countable $\Ax$-models.

\end{thm}

\begin{proof}
  As usual, we show that every
  $(\Log\Rules\Ax+\Name)$-$\Phi$-consistent set
  $\Psi\subseteq\Form(\Lambda)$ is satisfiable in a countable
  $\Ax$-model $M$ such that $M\models\Phi$ (where satisfiability is
  clearly invariant under passing from $\Form(\Lambda)$ to
  $\Form(\Lambda)^+$). The extended Lindenbaum lemma yields a
  $0$-pasted maximally consistent subset ABox $K \subseteq
  \Form(\Lambda)^+$ and a nominal $i$ in $\Form(\Lambda)^+$ such that
  $@_i \Psi \subseteq K$. By the named existence lemma, we find a
  named, hence countable, canonical $K$-model $M = (S_K, \gamma,
  V_K)$, and by the truth lemma (Lemma~\ref{lemma:ntl}), $M$ is an
  $\Ax$-model, $M\models\Phi$, and $M,K_i\models\Psi$.
\end{proof}

\begin{rem}
  In the literature (e.g.\ \cite[Theorem 7.29]{BlackburnEA01}), the
  above completeness theorem is sometimes phrased as ``completeness
  with respect to named models'', i.e. models where every state is the
  denotation of some nominal; such models also played a central role
  in the early development of hybrid logic by the Sofia school (see
  e.g.~\cite{PassyTinchev85}). In detail, this means that every state
  of the model is the denotation of a nominal \emph{in a language
    extended with countably many new nominals}. This extension is
  necessary, as otherwise the consistent set $\lbrace \neg n \mid n
  \in \Nom \rbrace$ would be satisfiable in a model where every state
  is named by a nominal $n \in \Nom$ of the original language, which
  is clearly impossible. Completeness with respect to models where
  every state is named by a nominal in an extended language, on the
  other hand, is an immediate consequence of completeness with respect
  to countable models.
\end{rem}
\begin{exa}
  The previous theorem establishes strong completeness results for
  pure extensions of all hybrid logics with neighbourhood semantics
  (Example~\ref{expl:logics}.\ref{item:classical}) that are defined by
  rank-1 axioms~\cite{SchroderPattinson08b}, i.e.\ modal formulas
  where the nesting depth of modalities is uniformly equal to $1$
  (such as the monotonicity axiom $\Box(a\land b)\modimpl\Box b$). For
  the monotonic cases, i.e.\ extensions of monotonic hybrid logic,
  these results are essentially known~\cite{tenCateLitak07}, while
  they seem to be new for the non-monotonic cases, i.e.\ extensions of
  classical hybrid logic not containing the monotonicity axiom,
  including, e.g., various deontic logics~\cite{Goble04}.  Moreover,
  the theorem newly proves strong completeness of the hybridization of
  coalition logic, as Theorem~3.2 of~\cite{Pauly02} essentially states
  that coalition logic satisfies strong one-step completeness.
\end{exa}

\subsection{Pure Completeness for Bounded Logics}
\label{sec:bounded}

\noindent The condition of strong one-step completeness used in the
previous section is a comparatively rare
phenomenon~\cite{SchroderPattinson08b}; the strength of the condition
becomes clear in the fact that, unlike in the classical case of Kripke
semantics, the above did not require a notion of
$1$-pastedness~\cite{BlackburnTzakova99}. We proceed to present an
alternative approach for the case where one does have an analogue of
the ($\mathsf{Paste}$-1) rule --- this is the case if the operators
are \emph{bounded}, i.e., their satisfaction hinges, in each case, on
only finitely and boundedly many states of a model. 
\begin{defi}
  A modal operator $\hearts$ is \emph{$k$-bounded} for $k \in \Nat$ with respect to a
	$\Lambda$-structure $T$ if
  for every set $X$ and every $A\subseteq X$,
  \begin{equation*}
    \Sem{\hearts}_X(A)=\textstyle\bigcup_{B\subseteq A,\#B\le k}\Sem{\hearts}_X(B).
  \end{equation*}
  (This implies in particular that $\hearts$ is monotonic.)  We say that
  $\Lambda$ is bounded w.r.t. $T$ if every modal operator $\hearts$ in
	$\Lambda$ is
  $k_\hearts$-bounded for some $k_\hearts$.
\end{defi}
\noindent
The boundedness of an operator can now be internalized in the logical
deduction system. In particular, for $k$-bounded operators $\hearts$, one
has the \emph{paste rule}
\begin{equation*}
  \lrule{\Paste{\hearts}(k)}{@_{j_1}\phi\land \dots\land @_{j_{k}}\phi\land 
    @_i \hearts(j_1\lor\dots\lor j_{k})\modimpl\psi}
  {@_i\hearts\phi\modimpl\psi}
\end{equation*}
with the side condition that the $j_r$ are pairwise distinct fresh
nominals. We write $\Phi\entails_{\Log\Rules\Ax + \Name + \Paste}
\phi$ if $\phi$ is derivable from assumptions in $\Phi$ in the system
$\Log\Rules+\Name$ where additionally the rule $(\Paste \hearts(k))$
may be used in deductions for $k$-bounded operators $\hearts$. This induces
the notion of $(\Log\Rules\Ax + \Name + \Paste)$-$\Phi$-consistency,
which we briefly refer to as consistency as we fix $\Phi$, $\Ax$, and
$\Rules$ throughout. Again, the system is clearly sound, i.e. $\Phi;
\Psi \models^\C \phi$ whenever $\Phi; \Psi \entails_{\Log\Rules\Ax +
  \Name+\Paste} \phi$, where $\C$ is the class of $\Ax$-models.

\begin{exas}\label{expl:bounded}
  \begin{sparenumerate}
  \item \emph{Hybrid $K$.} The modal operator $\Diamond$ is
    $1$-bounded. The arising paste rule $(\mathsf{Paste}\Diamond(1))$
    is precisely the rule $(\mi{paste}\Diamond)$
    of~\cite{BlackburnTenCate06}.
  \item \emph{Graded hybrid logic.} The modal operator $\gldiamond{k}$
    is $(k+1)$-bounded. One thus has a paste rule
    \begin{equation*}
      \lrule{\mathsf{Paste}\gldiamond{k}(k+1)}{@_{j_1}\phi\land \dots\land @_{j_{k+1}}\phi\land 
        @_i \gldiamond{k}(j_1\lor\dots\lor j_{k+1})\modimpl\psi}
      {@_i\gldiamond{k}\phi\modimpl\psi}
    \end{equation*}
    with side conditions as before.
  \item \emph{Positive Presburger hybrid logic.} A Presburger operator
    $\sum a_i\cdot\#(\argument_i)\ge k$ (Example~\ref{expl:logics}) is
    $k$-bounded if the $a_i$ are positive. E.g., this still allows
    expressing the statement, generally believed to be valid in the
    German national football league, that a team that has at least
    $37$ points will not be relegated:
    $3\cdot\#\mathsf{win}+1\cdot\#\mathsf{draw}\ge
    37\modimpl\neg\mathsf{relegated}$.
  \end{sparenumerate}
\end{exas}
\noindent
%
%
\noindent 
The generalized $1$-pastedness condition for bounded operators is as follows.
\begin{defi}
  Let $\Lambda$ be bounded. An ABox $K$ is \emph{$1$-pasted} if
  whenever $\hearts$ is $k$-bounded and $@_i \hearts \phi \in K$, then $\lbrace
  @_{j_1} \phi, \dots, @_{j_k} \phi, @_i \hearts(j_1 \lor \dots \lor j_k)
  \rbrace \subseteq K$ for some nominals $j_1, \dots, j_k$. 
\end{defi}
\noindent
Again, it is clear that if $\Lambda$ is bounded, then $K$ can induce a
named model only of $K$ is $1$-pasted. It is easy to see that if
$\Rules$ is one-step complete and $\Lambda$ is bounded (in fact
already if $\Rules$ derives monotony for every $\hearts\in\Lambda$), then
every $1$-pasted set is also $0$-pasted
(Definition~\ref{def:paste-zero}).
\begin{lem}[Extended Lindenbaum lemma for $1$-pasted sets]\label{lemma:ell-1}
  Let $\Lambda$ be bounded. If $\Psi \subseteq \FLang(\Lambda)$ is
  consistent, then there exist a $1$-pasted maximally consistent ABox $K
  \subseteq @\Form(\Lambda)^+$ and a nominal $i$ in $\Form(\Lambda)^+$
  such that $@_i\Psi \subseteq K$, where $\Form(\Lambda)^+$ is as in
  Section~\ref{sec:sosc}.
\end{lem}
\noindent
Bounded operators now allow us to use a weaker version of one-step
completeness.  Instead of requiring that all one-step consistent sets
are one-step satisfiable, we may restrict to \emph{finite extensions}
of propositional variables.
\begin{defi}
  We say that $\Rules$ is \emph{strongly finitary one-step complete} if
  for every set $X$, every one-step consistent subset of
  $\Prop(\Lambda(\Pfin(X)))$ is one-step satisfiable.
\end{defi}
\noindent
Clearly, any strongly one-step complete rule set is also strongly
finitary one-step complete, but the example of graded hybrid logic
witnesses that the converse is not true. We note that the weaker
criterion still fails for probabilistic logics due to inherent
non-compactness~\cite{SchroderPattinson09}; probabilistic logics also
fail to be bounded, as a given probability $p\in[0,1]$ can be split
into any number of summands. Together with boundedness, the above
condition enables a second version of the named existence lemma.
\begin{lem}[Named existence lemma, Version 2]\label{lem:nel2}
  If $\Lambda$ is bounded, $\Rules$ is strongly finitary one-step
  complete, and $K$ is $1$-pasted, then there exists a named canonical
  $K$-model.
\end{lem}
\noindent
Summarizing the above, we have the following extended completeness
result.
\begin{thm}\label{thm:pure-ext-bounded}
  Let $\Lambda$ be bounded, and let $\Rules$ be strongly finitary
  one-step complete.  Then every extension of $\Log\Rules$ by pure
  axioms is globally and locally strongly complete over countable
  hybrid models when equipped with the $\Name$ and $\Paste$ rules.  In
  other words, if $\Phi, \Psi \subseteq \Form(\Lambda)$,
  $\phi\in\Form(\Lambda)$, and $\C$ is the class of all countable
  $\Ax$-models, then $\Phi; \Psi \entails_{\Lang\Rules\Ax + \Name +
    \Paste} \phi$ whenever $\Phi; \Psi \models^\C \phi$.
\end{thm}
\noindent
The proof follows the same route via extended Lindenbaum lemma,
existence lemma, and truth lemma as for
Theorem~\ref{thm:pure-ext-mcs}.
\begin{exa}
  By Example~\ref{expl:bounded} and the fact that the known complete
  axiomatizations of the associated modal logics are in fact strongly
  finitary one-step complete, the previous theorem proves completeness
  of pure extensions of hybrid $K$, graded hybrid logic, and positive
  Presburger hybrid logic. Except for the standard case of hybrid $K$,
  these results seem to be new. In particular, we obtain completeness
  of pure extensions of graded (or positive Presburger) hybrid logic
  defining the following frame classes in multigraph semantics:
  \begin{itemize}
  \item The class of \emph{Kripke frames}, seen as the class of
    multigraphs where the transition multiplicity between two
    individual states is always at most $1$, defined by the pure axiom
    $\neg\gldiamond{1}i$.
    \item The class of \emph{reflexive} multigraphs,  defined by the pure
      axiom $i\modimpl\gldiamond{0}i$.
    \item The class of \emph{transitive} multigraphs, defined by the
      pure axioms $\gldiamond{0}\gldiamond{n}i\modimpl\gldiamond{n}i$,
      $n\ge 0$.
    \item The class of \emph{symmetric} multigraphs, i.e., those where
      the transition multiplicity from $x$ to $y$ always equals the
      one from $y$ to $x$, which is defined by the pure axioms
      $i\land\gldiamond{k}j\modimpl @_j\gldiamond{k}i$.
    \end{itemize}
  Other frame classes of interest, see e.g. \cite[Section
	7.3]{BlackburnEA01}, can be characterized similarly by translating
	the corresponding frame conditions from Kripke to multigraph
	semantics.
\end{exa}

\subsection{The Mixed Case}

\noindent In some cases, the two methods laid out in the preceding
sections can be combined for modal operators with several arguments
that adhere, in each of their arguments, to one of the respective sets
of semantic conditions. For the sake of readability, we formulate this
explicitly only for the mixed binary case with a single modal
operator, i.e.\ we assume in this section that $\Lambda=\{\hearts\}$ with
$\hearts$ binary; the generalization to arbitrary numbers of arguments,
several modal operators etc.\ should be obvious, and essentially only
requires more elaborate terminology and notation.
\begin{defi}
  We say that $\Rules$ is \emph{(strongly, strongly finitary) one-step
    complete} if every one-step consistent subset of
  $\Prop(\Lambda(\Pow(X)\times\Pfin(X)))$ is one-step satisfiable.
  Moreover, we say that $\hearts$ is \emph{$k$-bounded in the second
    argument} for $k\in\Nat$ if for every set $X$ and all
  $A,B\subseteq X$,
  \begin{math}
    \Sem{\hearts}_X(A,B)=\bigcup_{C\subseteq A, \#C\le k}\Sem{\hearts}_X(A,C).
  \end{math}
\end{defi}
\noindent In the same manner as for Theorems~\ref{thm:pure-ext-mcs}
and~\ref{thm:pure-ext-bounded}, we derive:
\begin{thm}\label{thm:pure-ext-mixed}
  If $\Rules$ is (strongly, strongly finitary) one-step complete and
  $\hearts$ is $k$-bounded in the second argument, then every extension of
  $\Log\Rules$ by pure axioms is both locally and globally strongly
  complete over countable hybrid models when equipped with the
  appropriate $\Name$ and $\Paste$ rules.
\end{thm}
\begin{exa}
  Hybrid $\CK$ (Example~\ref{expl:logics}) is easily seen to be
  (strongly, strongly finitary) one-step complete, and the operator
  $>$ defined from the conditional operator $\CondArrow$ by
  $a>b:\modiff \neg(a\CondArrow \neg b)$ is $1$-bounded in the second
  argument. By the above, it follows that every pure extension of
  hybrid $\CK$ is strongly complete over countable hybrid selection
  function models. E.g.\ we may define the class of conditional frames
  where all expressible conditions induce transitive relations by pure
  axioms $(\phi>\phi>i)\modimpl(\phi>i)$.  Such frames satisfy also
  the dual axiom (using a propositional variable $a$) $(\phi\CondArrow
  a)\modimpl(\phi\CondArrow(\phi\CondArrow a))$, an axiom for
  duplicating conditional assumptions. Similar statements apply to a
  combination of graded and conditional logic (obtainable
  compositionally using the methods of~\cite{SchroderPattinson07}),
  which has operators of the form $a\CondArrow_k b$ ``if $a$, then one
  normally has more than $k$ instances of $b$''. 

  The semantics of conditional logics in general has complex
  ramifications, involving, e.g., preference orderings or systems of
  spheres (see, e.g.,~\cite{FriedmanHalpern94,Sano07}); application of
  our methods to conditional logics beyond $\CK$ is the subject of
  further investigation. We note that pure completeness of a hybrid
  extension of Lewis' logic of counterfactuals has been established
  recently~\cite{Sano07}.
\end{exa}

\section{Local Binding}

\noindent We next investigate completeness of a stronger hybrid
language that includes the $\downarrow$ binder, which binds a state
variable to the current state. Concretely, we allow formulas of the
form $\downarrow i.\,\phi$, wherein the nominal $i$ is locally bound
(for compactness of presentation, we give up the usual distinction
between nominals and state variables). Given a modal similarity type
$\Lambda$, we write $\Form_\downarrow(\Lambda)$ for the ensuing
extension of $\Form( \Lambda)$. The reading of the formula $\downarrow
i. \phi$ is ``$\phi$ holds for the current state $i$''. The
satisfaction relation in the extended logic is defined by an
additional clause for the $\downarrow$ binder,
\begin{equation*}
  (C,\gamma,V)\models \downarrow i.\,\phi\textrm{ iff }(C,\gamma,V[c/i])\models\phi
\end{equation*}
where $c$ is a state in a coalgebra $C$ and $V[c/i]$ is obtained from
$V$ by modifying the value of $i$ to $c$.  The semantics of the
$\downarrow$ binder immediately translates into the axiom scheme (see
e.g.~\cite{BlackburnTenCate06})
\begin{equation*}
  (\mathsf{DA})\qquad @_i((\downarrow j.\,\phi)\modiff\phi[i/j]).
\end{equation*}
Given a set $\Rules$ of $\Lambda$-rules, a set $\Phi \subseteq \Form_\downarrow(\Lambda)$ of formulas
and a set $\Ax \subseteq @\Form_\downarrow(\Lambda)$ of pure axioms,
we write $\Phi \entails_{\Log\Rules\Ax + \Name+\Paste + \mathsf{DA}}
\phi$
for the
extension of the associated provability predicate
$\entails_{\Lang\Rules\Ax + \Name + \Paste}$ with $(\mathsf{DA})$.
Using $(\mathsf{DA})$, one easily proves an
extension of the truth lemma for named models (Lemma~\ref{lemma:ntl})
to $\Form_\downarrow(\Lambda)$, so that the completeness results for pure
extensions proved before
(Theorems~\ref{thm:pure-ext-mcs},~\ref{thm:pure-ext-bounded},
and~\ref{thm:pure-ext-mixed}) transfer immediately to
$\Lang_\downarrow$. We make this explicit for the bounded case:
\begin{thm}\label{thm:pure-ext-binder}
  If $\Lambda$ is bounded and $\Rules$ is strongly finitary one-step
  complete, then every pure extension of $\Lang_\downarrow$ is
  strongly locally and globally complete over countable hybrid
  models. In other words, $\Phi; \Psi \models^\C \phi$ iff $\Phi; \Psi
  \entails_{\Lang\Rules\Ax + \Name + \Paste + \mathsf{DA}} \phi$ for
  all $\phi \in \Form_\downarrow(\Lambda)$ and all
  $\Phi,\Psi\subseteq\Form(\Lambda)$, where $\C$ is the class of all
  countable $\Ax$-models.
\end{thm}
%

\begin{rem}
  As noted in~\cite{tenCateLitak07}, the named model construction more
  generally yields completeness for any \emph{locally definable}
  extension of the hybrid language, i.e.\ any extension whose
  semantics \emph{at named states} is defined by a formula similar to
  $(\mathsf{DA})$.
\end{rem}

\begin{exa}
  Continuing Example~\ref{expl:bounded},
  Theorem~\ref{thm:pure-ext-binder} reproves not only the known
  completeness of pure extensions of hybrid $K$ with $\downarrow$, but
  also the completeness of pure extensions of graded (or positive
  Presburger) hybrid logic with $\downarrow$. This extends easily
  to the multi-agent case, or, in description logic terminology, to
  description logics with multiple roles. As, moreover, both a role
  hierarchy and transitivity of roles can be defined using pure
  axioms, we thus arrive at a complete axiomatization of an extension
  of the description logic $\mathcal{SHOQ}$ with satisfaction
  operators and $\downarrow$, which has been used in connection with
  conjunctive query answering~\cite{GlimmEA06}, and allows, e.g.,
  talking about the number of stepchildren of a stepmother, in
  continuation of the stepmother example from~\cite{Marx02}, .
\end{exa}

\section{Conclusions}

\noindent We have laid out two criteria for the existence of named
canonical models in coalgebraic hybrid logics --- one that applies to
cases where one has an analogue of the so-called Paste-$1$ rule of
standard hybrid logic, and one which applies to cases where one does
not need any such rule. While the latter means essentially that the
logic is equipped with a neighbourhood semantics, the former requires
that all modal operators of the logic are bounded, i.e.\ there is
always only a bounded number of states relevant for their satisfaction
at each point. Our main novel example of this type is graded hybrid
logic (and an extension of it using certain Presburger
modalities~\cite{DemriLugiez06}). The named model construction entails
completeness of pure extensions and completeness of extended hybrid
languages with the local binder $\downarrow$ (of which the I--me
construct of~\cite{Marx02} is a single-variable restriction), which we
thus obtain as new results for, e.g., hybrid coalition logic, hybrid
classical modal logic, several hybrid deontic logics, hybrid
conditional logic, graded hybrid logic, and an extension of the
description logic $\mathcal{SHOQ}$. An open question that remains is
the existence of so-called orthodox
axiomatizations~\cite{BlackburnTenCate06} in the presence of
$\downarrow$, as well as to find an analogue of the characterization
result of~\cite{tenCateLitak07} stating that a variant of the
Paste-$1$ rule characterizes the Kripke models among the topological
models of $S4$. A further topic of investigation is to find decidable
fragments of the language with $\downarrow$; we note slightly
speculatively that the fragment used in~\cite{Marx02} may, in our
terminology, be seen as requiring that a suitably defined NNF of a
formula contains only positive occurrences of bound nominals under
bounded modal operators.

 \bibliographystyle{myabbrv} \bibliography{schroder}

\providecommand{\noopsort}[1]{}
\begin{thebibliography}{10}

\bibitem{Areces99}
C.~Areces, P.~Blackburn, and M.~Marx.
\newblock A road-map on complexity for hybrid logics.
\newblock In {\em Computer Science Logic, CSL 99}, vol. 1683 of {\em LNCS}, pp.
  307--321, 1999.

\bibitem{BaaderEA03}
F.~Baader, D.~Calvanese, D.~L. McGuinness, D.~Nardi, and P.~F. Patel-Schneider,
  eds.
\newblock {\em The Description Logic Handbook}.
\newblock Cambridge University Press, 2003.

\bibitem{BlackburnEA01}
P.~Blackburn, M.~de~Rijke, and Y.~Venema.
\newblock {\em Modal Logic}.
\newblock Cambridge University Press, 2001.

\bibitem{BlackburnTenCate06}
P.~Blackburn and B.~ten Cate.
\newblock Pure extensions, proof rules, and hybrid axiomatics.
\newblock {\em Stud.\ Log.}, 84:277--322, 2006.

\bibitem{BlackburnTzakova99}
P.~Blackburn and M.~Tzakova.
\newblock Hybrid languages and temporal logic.
\newblock {\em Logic J.\ IGPL}, 7:27--54, 1999.

\bibitem{Chellas80}
B.~Chellas.
\newblock {\em Modal Logic}.
\newblock Cambridge University Press, 1980.

\bibitem{CirsteaEA09}
C.~Cirstea, A.~Kurz, D.~Pattinson, L.~Schr{\"o}der, and Y.~Venema.
\newblock Modal logics are coalgebraic.
\newblock {\em The Computer Journal}, 2009.
\newblock In print.

\bibitem{DAgostinoVisser02}
G.~D'Agostino and A.~Visser.
\newblock Finality regained: A coalgebraic study of {Scott}-sets and multisets.
\newblock {\em Arch.\ Math.\ Logic}, 41:267--298, 2002.

\bibitem{DemriLugiez06}
S.~Demri and D.~Lugiez.
\newblock {P}resburger modal logic is only {PSPACE}-complete.
\newblock In {\em Automated Reasoning, IJCAR 06}, vol. 4130 of {\em LNAI}, pp.
  541--556. Springer, 2006.

\bibitem{FriedmanHalpern94}
N.~Friedman and J.~Y. Halpern.
\newblock On the complexity of conditional logics.
\newblock In {\em Knowledge Representation and Reasoning, KR 94}, pp. 202--213.
  Morgan Kaufmann, 1994.

\bibitem{GlimmEA06}
B.~Glimm, I.~Horrocks, and U.~Sattler.
\newblock Conjunctive query answering for description logics with transitive
  roles.
\newblock In {\em Description Logics, DL 06}, vol. 189 of {\em CEUR Workshop
  Proceedings}. CEUR-WS.org, 2006.

\bibitem{Goble04}
L.~Goble.
\newblock A proposal for dealing with deontic dilemmas.
\newblock In {\em Deontic Logic in Computer Science, DEON 04}, vol. 3065 of
  {\em LNAI}, pp. 74--113. Springer, 2004.

\bibitem{Marx02}
M.~Marx.
\newblock Narcissists, stepmothers and spies.
\newblock In {\em Description Logics, DL 02}, CEUR Workshop Proceedings, 2002.

\bibitem{MyersEA09}
R.~Myers, D.~Pattinson, and L.~Schr{\"o}der.
\newblock Coalgebraic hybrid logic.
\newblock In {\em Foundations of Software Science and Computation Structures,
  FOSSACS 09}, LNCS. Springer, 2009.
\newblock To appear.

\bibitem{PassyTinchev85}
S.~Passy and T.~Tinchev.
\newblock {PDL} with data constants.
\newblock {\em Inf.\ Process.\ Lett.}, 20:35--41, 1985.

\bibitem{Pattinson03}
D.~Pattinson.
\newblock Coalgebraic modal logic: Soundness, completeness and decidability of
  local consequence.
\newblock {\em Theoret.\ Comput.\ Sci.}, 309:177--193, 2003.

\bibitem{Pauly02}
M.~Pauly.
\newblock A modal logic for coalitional power in games.
\newblock {\em J.\ Logic Comput.}, 12:149--166, 2002.

\bibitem{Sano07}
K.~Sano.
\newblock Hybrid counterfactual logics.
\newblock {\em J.\ Log.\ Lang.\ Inf.}, 18:515–539, 2009.

\bibitem{Schroder07}
L.~Schr{\"o}der.
\newblock A finite model construction for coalgebraic modal logic.
\newblock {\em J.\ Log.\ Algebr.\ Prog.}, 73:97--110, 2007.

\bibitem{SchroderPattinson08b}
L.~Schr{\"o}der and D.~Pattinson.
\newblock Rank-1 modal logics are coalgebraic.
\newblock {\em J.\ Logic Comput.}
\newblock In print.

\bibitem{SchroderPattinson07}
L.~Schr{\"o}der and D.~Pattinson.
\newblock Modular algorithms for heterogeneous modal logics.
\newblock In {\em Automata, Languages and Programming, ICALP 07}, vol. 4596 of
  {\em LNCS}, pp. 459--471. Springer, 2007.

\bibitem{SchroderPattinson08}
L.~Schr{\"o}der and D.~Pattinson.
\newblock {PSPACE} bounds for rank-1 modal logics.
\newblock {\em ACM Trans.\ Comput.\ Log.}, 10(2:13):1--33, 2009.

\bibitem{SchroderPattinson09}
L.~Schr{\"o}der and D.~Pattinson.
\newblock Strong completeness of coalgebraic modal logics.
\newblock In {\em Theoretical Aspects of Computer Science, STACS 09}, Leibniz
  International Proceedings in Informatics, pp. 673--684. Schloss Dagstuhl --
  Leibniz-Zentrum f{\"u}r Informatik, 2009.

\bibitem{tenCateLitak07}
B.~ten Cate and T.~Litak.
\newblock Topological perspective on the hybrid proof rules.
\newblock In {\em Hybrid Logic, HyLo 06}, vol. 174 of {\em ENTCS}, pp. 79--94,
  2007.

\end{thebibliography}

\end{document}